\documentclass{eptcs}
\usepackage{oldlfont,amssymb,latexsym,euscript,epsf,makeidx}
\usepackage[dvips]{graphicx}
\input{xy}
\xyoption{all}
\def\[{\ensuremath{[ \! [}}
\def\]{\ensuremath{] \! ]}}
\def\coprod{\bigoplus}

\def\sqin{\ensuremath{\sqsubset \!\!\!\!\!\! - }} 
\def\C{{\cal{C}}}
\def\S{{\cal{S}}}
\def\G{{\cal{G}}}
\def\L{{\cal{L}}}

\def\lra{\longrightarrow}

\def\0{0}

\def\1m{\iota}
\def\restrict{\!\! \upharpoonright \!\!}
\def\PCF{{\mathrm{PCF}}}

\def\1{1}

\def\E{{\cal E}}
\def\restrict{\!\! \upharpoonright \!\!}
\def\PCF{{\mathrm{PCF}}}

\newcommand{\fst}{{\mathtt{fst}}}
\newcommand{\snd}{{\mathtt{snd}}}
\newcommand{\thro}{{\mathsf{throw}}}
\newcommand{\xs}{{\mathsf{S}}}
\newcommand{\loc}{{\mathsf{loc}}}
\newcommand{\ex}{{\mathsf{Ex}}}
\newcommand{\throwe}{{\mathtt{throw_e}}}

\newcommand{\cse}{{\mathtt{case}}}
\newcommand{\as}{{\mathtt{as}}}
\newcommand{\cont}{{\mathsf{cont}}}

\newcommand{\void}{{\mathtt{void}}}
\newcommand{\match}{{\mathtt{match}}}
\newcommand{\Loc}{\mathsf{Loc}}

\newcommand{\wb}{{\G_{WB}}}

\newcommand{\cwb}{{\cg_{WB}}}

\newcommand{\wkb}{{\mathsf{wkb}}}

\newcommand{\obs}{{\mathrm{obs}}}

\newcommand{\emp}{{{0}}}

\newcommand{\CG}{{\cal{CG}}}

\newcommand{\dom}{{\mathrm{dom}}}

\newcommand{\lbl}{{\mathsf{label}}}
\newcommand{\throw}{{\mathtt{throw}}}

\newcommand{\goto}{{\mathsf{goto}}}

\newcommand{\set}{{\mathtt{set}}}
\newcommand{\deref}{{\mathtt{deref}}}
\newcommand{\If}{{\mathtt{If}}}
\newcommand{\jump}{{\mathsf{jump}}}

\newcommand{\trm}{{\mathrm{Tm}}}
\newcommand{\prog}{{\mathrm{Prog}}}

\newcommand{\lett}{{\mathtt{let}}}
\newcommand{\with}{{\mathtt{with}}}

\newcommand{\pending}{{\mathsf{pending}}}
\newcommand{\ok}{{\mathsf{ok}}}
\newcommand{\caught}{{\mathsf{caught}}}
\newcommand{\store}{{\mathcal{S}}}

\newcommand{\open}{{\mathsf{open}}}
\newcommand{\lamstat}{\lambda_R}
\newcommand{\statcon}{\lambda_{RC}}
\newcommand{\lamex}{\lambda_E}
\newcommand{\lamcon}{\lambda_C}
\newcommand{\statex}{\lambda_{RE}}
\newcommand{\lcre}{\lambda_{RCE}}

\newcommand{\eqrc}{\simeq_{RC}}

\newcommand{\spc}{\hspace{2pt}}

\newcommand{\inj}{\mathrm{in}}
\newcommand{\inl}{{\mathtt{in_l}}}
\newcommand{\inr}{{\mathtt{in_r}}}

\newcommand{\callcc}{{\mathtt{callcc}}}

\newcommand{\nat}{{\mathtt{nat}}}
\newcommand{\Lr}{{\cal{L_R}}}
\newcommand{\Lc}{{\cal{L_{RC}}}}

\newcommand{\Lce}{{\cal{L_{RCE}}}}
\newcommand{\Llce}{{\cal{L_{CX}}}}
\newcommand{\xn}{\ {\mathtt{exn}}}
\newcommand{\nxn}{{\mathtt{new\_exn}}}
\newcommand{\thread}{{\mathsf{thread}}}
\newcommand{\threads}{{\mathsf{threads}}}
\newcommand{\unit}{{\mathtt{unit}}}
\newcommand{\refn}{{\mathtt{ref}}}

\newcommand{\mkxn}{{\mathtt{mkexn}}}

\newcommand{\var}{{\mathtt{var}}}
\newcommand{\pcf}{{\mathrm{PCF}}}
\newcommand{\callc}{{\mathsf{callcc}}}

\newcommand{\rais}{{\mathtt{raise} }}
\newcommand{\handle}{{\mathtt{handle}}}
\newcommand{\new}{{\mathtt{new}}}
\newcommand{\xcell}{{\mathsf{exn}}}
\newcommand{\cell}{{\mathsf{cell}}}
\newcommand{\raze}{{\mathsf{raise}}}

\newcommand{\ski}{[()]}
\newcommand{\handl}{{\mathbf{handle}}}

\newcommand{\T}{{\mathbf{T}}}

\newcommand{\Fam}{{\mathsf{Fam}}}

\newcommand{\catch}{{\mathtt{catch}}}
\newcommand{\com}{{{1}}}
\newcommand{\exn}{{\mathtt{exn}}}

\newcommand{\Cu}{{\EuScript C}}
\newcommand{\abort}{{\EuScript A}}
\newcommand{\arl}{{\longrightarrow}}

\newcommand{\ff}{{\mathbf{ff}}}
\newcommand{\tru}{{\mathtt{tt}}}

\newcommand{\then}{{\mathtt{then}}}

\newcommand{\el}{{{\mathtt{else}}}}

\newcommand{\id}{{\mathsf{id}}}

\newcommand{\app}{{\mathsf{app}}}

\newcommand{\inn}{{{\mathtt{in}}}}

\newcommand{\catchc}{{\mathtt{callcc}}}
\newcommand{\Lx}{{\cal{L}_{X}}}

\newcommand{\ter}{{\mathbf{1}}}

\newcommand{\R}{{{R}}}

\newtheorem{theorem}{Theorem}[section]
\newtheorem{lemma}[theorem]{Lemma}
\newtheorem{definition}[theorem]{Definition}

\newtheorem{proposition}[theorem]{Proposition}
\newtheorem{corollary}[theorem]{Corollary}

\newif\ifqed
\newcommand{\qed}{\global\qedfalse\noindent\unskip\penalty50\null\nobreak\hskip1em\nobreak\hfill$\Box$}
\def\titlerunning{Combining and Relating Control Effects and their Semantics}
\def\authorrunning{J. Laird}

\let\Box\undefined

\DeclareMathSymbol{\Box}{\mathord}{AMSa}{"03}
\newenvironment{proof}[1][Proof]{\global\qedtrue\trivlist\item[\hskip\labelsep\textsc{#1:}]}{\ifqed\qed\fi\endtrivlist}

\begin{document}
\title{Combining and Relating Control Effects and their Semantics }
\author{James Laird
\institute{Department of Computer Science, University of Bath, UK}}
\date{}
\maketitle
\begin{abstract}Combining local exceptions and first class continuations leads to programs with complex control flow, as well as the possibility of expressing powerful constructs such as resumable  exceptions.  We describe and compare games models for a  programming language which includes these features, as well as higher-order references. 
They are obtained by contrasting methodologies: by annotating sequences of moves with ``control pointers'' indicating where exceptions are thrown and caught, and by composing the  exceptions and continuations monads.

 The former approach allows an explicit representation of control flow in games for exceptions, and hence a straightforward proof of definability (full abstraction) by factorization, as well as offering the possibility of a semantic approach to control flow analysis of exception-handling. However, establishing soundness of such a concrete and complex model is a non-trivial problem. It may be resolved by establishing  a correspondence  with the monad semantics, based on erasing explicit exception moves and replacing them with control pointers.  
\end{abstract}

\section{Introduction}
Control effects such as exceptions and continuations are  key features of higher-order programming languages. They are typically used to recover from errors, and may result in complicated and unpredictable control flow in programs. Therefore, principles for reasoning about notions such as \emph{exception safety} are potentially useful and important. Denotational semantics provides one basis for such principles. Here, broadly speaking, there are two  approaches to describing computational effects.  Constructions such as \emph{monads}, and \emph{continuation-passing-style interpretations} yield  useful algebraic theories for reasoning \emph{soundly} about programs, although they impose additional layers of definition and interpretation through which reasoning about programs must be filtered, particularly in the presence of  properties such as locality. By contrast, \emph{game semantics} provides a framework in which to model combinations of effects more directly by the relaxation of constraints on strategies representing functional programs. This approach has  been used successfully to give \emph{fully abstract} interpretations of many features, including an   account of locality for  features such as state \cite{AHM}. However, the  combinatorial nature of games models means that reasoning about denotations  --- for example, proving basic soundness results --- can be  difficult in the absence of structuring principles.

Thus it can be useful to relate the direct (games) and indirect (monads, CPS) approaches to effects, to gain the advantages of both representations. This paper will do so for exceptions and continuations. In the process, we construct a first fully abstract model for a language which combines continuations and locally declared exceptions, as in Standard ML of New Jersey. Although many control structures can be implemented using either feature, exceptions and continuations exhibit several subtle but significant differences in behaviour: one way of understanding these is by studying the interaction of the two effects in combination. (For example, observing that exceptions break key equational rules which hold for continuations \cite{LE}.) 
Combining exceptions and continuations also provides a way of interpreting further, powerful control constructs: they may be used to macro-express \emph{resumable} exceptions, and implement dynamic delimited control operators such as \emph{prompts} \cite{GRR}.

Exceptions and continuations also provide a test case for semantic theories of  combining  algebraic effects,  studied in detail in \cite{HLPP}. Here we shall simply use the fact that  there is a distributive law of the monad $\_ + E$ (exceptions) over  $\R^{\R^{\_}}$ (continuations)  (which exist for objects  $E$ and $\R$ whenever the relevant categorical constructions do), since the exceptions monad  distributes over any other monad.  Thus we have an exceptions-and-continuations monad $\R^{\R^{\_ +E}}$.

How can this monad  be related to a game semantic account?  In the case of  first-class continuations (on their own), there is a simple correspondence between the games and monadic interpretations --- relaxing the \emph{well-bracketing condition} on strategies renders  the lifted sum  monad $\Sigma\_$ introduced in \cite{AMV}  isomorphic to the continuations monad $\R^{\R^\_}$, where $\R$ is the  ``one-move game'', giving both direct and indirect (continuation-passing style) interpretations of  call/cc \cite {LT}.

The case of exceptions is more complicated.  We may interpret a single global exception by adding to our games distinguished ``exception answer'' moves for each question. Extending the continuations monad with such an answer yields a monad formally equivalent to $\R^{\R^{\_ + \1}}$\footnote{Another approach in \cite{ntex} also uses an exceptions monad on a category of nominal games --- here we focus on more concrete models with implicit state.} In the presence of local state, this is sufficient to macro-express  local exception declaration, as we may use imperative variables both as flags to indicate which exception has been set, and to carry exceptional values. However, this leaves open the problem of identifying the elements of the model definable using local exception handling, and their intrinsic equivalence.

Exceptions have also been more difficult to incorporate into the simple picture of relaxing constraints on strategies to get more powerful effects;  locally declared exceptions can be interpreted directly by relaxing the bracketing condition to a ``weak bracketing'' condition \cite{Lli}, but fully capturing this behaviour also  requires new information to be added to strategies in the form of additional ``control pointers'' attached to sequences. 
 Relaxing the weak bracketing condition also gives a  straightforward and intuitively natural alternative denotation for call-with-current-continuation in the context of exception handling --- playing a control pointer to  a ``closed'' move allows the handler-context to be reset. However, this representation of continuations and exceptions is rather implicit, and does not lend itself to reasoning about equivalence between denotations of programs --- even to the limited extent of  proving soundness with respect to the operational semantics.

The solution adopted here is a  correspondence with the exceptions monad and CPS interpretations, given by relating exception-arenas to control games by replacing exception moves with control pointers (in this case, indicating which question is \emph{pending} when an exception is thrown).  Finally, we prove  full abstraction results for the control games models using \emph{factorization} into the model with only local control defined in \cite{AHM}.

\section{An Effectful Functional Programming  Language} 
We shall first describe a simply-typed call-by-value programming language $\Lce$ with (locally declared) general references,  first-class continuations and local exceptions (which might be considered as a simply-typed fragment of SML of New Jersey). The core of the language, $\L$, is a simply-typed call-by-value $\lambda$-calculus based on the computational $\lambda$-calculus \cite{MogC}.

  \emph{Types} of $\L$ are generated  from the product, sum (and their units  $\com$ and $\emp$) and function types:\\
$S,T: = \emp\ |\ \com\ |\ S \times T\ |\  S + T\  |\ S \rightarrow T$\\
We distinguish \emph{computation} and \emph{value} terms.
\emph{Values} are given by the grammar:\\
\begin{tabular}{l l}
$U,V:=$ &$ x\ |\ () \ |\ \langle U,V\rangle\ |\ \inj_1(V)\ |\ \inj_2(V)\ |\  \lambda x.M $\\

\end{tabular}\\ 
\emph{Computations} are given by the grammar:\\
\begin{tabular}{l l}
$M,N:=$ &$ [V]\ |\ \lett\spc x= M\spc \inn \spc N\ |\  \void\spc V\ |\ U\spc V\ |\ \match\spc V \spc \as\spc (x,y). M\ |\  \cse\spc V\spc \as \spc \inj_1(x).M| \inj_2(x).N $\\

\end{tabular}\\
 Typing judgements, of the form $\Gamma \vdash_c M:T$ for computations, and $\Gamma \vdash_v V:T$ for values,  are given in Table 1. We write $M;N$ for $\lett \spc x=M\spc\inn\spc  N$, if $x$ is not free in $M$ or $N$. 

\subsection{Computational Effects}
Computational effects  are introduced by adding constructs for declaring references and exceptions, and capturing the current continuation as a first-class function, as follows: 

\paragraph{References} The type $\var[T]$ of references to values of type $T$ is \emph{defined}  to be $(T \rightarrow \com) \times (\com \rightarrow T)$ --- the product of the types of its methods, assignment and dereferencing, which may be recovered by left and right projection, respectively ---  i.e. given $a:\var[T]$ and $V:T$, we sugar $\match \spc a\spc\as\spc (x,y).x\spc V$ as $a:= V$, and $\match \spc a \spc \as\spc (x,y). y\spc ()$ as $\deref(a)$.

Thus the only further  syntax we need to add to our type theory  is a constant (value) $\new:\com \rightarrow \var[T]$ for declaring a new reference. We write     $\lett\spc x\spc = (\new\spc ()) \spc \inn\spc  x:=V;M$ as $\new\spc x:=V.M$. 

\paragraph{Exceptions} The type of $\exn$ of exceptions is similarly defined to be the product  $((\com \rightarrow \emp) \rightarrow \com) \times (\com \rightarrow \emp)$  of its method types: \emph{throwing} of type $\com \rightarrow \emp$ and \emph{catching}, of type $(\com \rightarrow \emp) \rightarrow \com$.\footnote{The ``thunked'' empty type $\com \rightarrow \emp$ is used to represent the type of computations which do not return a value.}  Given $e:\exn$, we sugar  $\match \spc e \spc \as\spc (x,y).x \spc \lambda ().N$  and $\match \spc e\spc \as (x,y).y\spc ()$ as $\catch\spc e\spc \inn \spc N$ and $\throw(e)$, respectively. 

Thus to extend our type theory with exceptions it is sufficent to add a value $\nxn:\com \rightarrow \exn$ for declaring a new exception. 
  
\paragraph{Continuations} As in New Jersey SML, we introduce first-class continuations via a value $\callcc:((T \rightarrow S) \rightarrow T) \rightarrow T$, which passes a first class representation of the current \emph{continuation} (as a value of type $T \rightarrow S$ for arbitrary $S$) to its argument.

\begin{table}
\begin{center}
\begin{tabular}{c c c} 
{\Large $\over \Gamma,x:T \vdash_v x:T$} & {\Large $ \Gamma \vdash_v V:T \over \Gamma \vdash_c [V]:T$}  & {\Large $\Gamma \vdash_c M:S\ \ \ \ \Gamma,x:S \vdash_c N:T \over \Gamma \vdash_c \lett \spc x = M\spc \inn\spc N:T$}\\\\
{\Large $\over \Gamma \vdash_v ():1$}&  & {\Large $\Gamma \vdash_v V:0 \over \Gamma \vdash_c \void\spc V: T$} \\\\
{\Large $\Gamma \vdash_v U:S \ \ \ \ \Gamma \vdash_v V:T  \over \Gamma \vdash_v \langle U,V\rangle :S \times T$} & & {\Large $\Gamma \vdash_v V:S \times S' \ \ \ \ \Gamma,x:S,y:S' \vdash_c M:T \over  \Gamma \vdash_c \match\spc V\spc \as\spc (x,y).M:T$}  \\\\
{\Large $\Gamma \vdash_v U:S \over \Gamma \vdash_v \inj_i(V):T_1 + T_2$}\mbox{$i \in \{1,2\}$}& & {\Large $\Gamma \vdash_v V:S + S' \ \ \ \Gamma,x:S \vdash_c M:T\ \ \  \Gamma,x:S' \vdash_c N:T \over \Gamma \vdash_c \cse\spc V\spc \as\spc \inl(x).M|\inj_2(x).N:T$}\\\\
{\Large $\Gamma, x:S \vdash_c M:T \over \Gamma \vdash_v \lambda x.M:S \rightarrow T$} & & {\Large $\Gamma \vdash_v U:S\rightarrow T\ \ \ \ \Gamma \vdash_v V:S \over \Gamma \vdash_c U \spc V:T$}\\\\
\end{tabular}
\caption{Typing Judgements for Computations and Values}
\end{center}
\end{table}
We make use of the following fragments of $\Lce$ --- the purely functional fragment $\L$,  the fragment  $\Lr$ with local control  (i.e. references but no continuations or exceptions, omitting  the constants $\nxn$ and $\callcc$: this is essentially the language   defined  in \cite{AHM}, with its games model), and   the fragment $\Lc$ with continuations and references but no exceptions. 
\subsection{Operational Semantics}
To give an operational semantics for $\Lce$, we introduce constants representing the capacity to read from and write to a location, and raise and handle an exception, and  a new constructor, representing composition with the top-level continuation. 
Let ${\mathcal{L_{CE}^{\#}}}$ be the extension of $\Lce$ with:
\begin{itemize}
\item An unbounded set of pairs of constants $(\set(a),!a)$.
\item An unbounded set of pairs of constants $(\throw(e),\catch(e))$. 
\item An operation $\#\_$ taking computations of type $\com$ to computations of type $T$. 

\end{itemize}
\emph{Evaluation contexts} $E[\_]$ are given by the grammar:
$$E[\_]::= [\_]\  |\ \lett\spc x = E[\_]\spc \inn \spc M\ |\  \catch(e)\spc \lambda x.E[\_]$$ 
$E_h[\_]$ denotes an evaluation context without a $\catch(h)\spc \lambda x.\_$  in the spine --- i.e. given by the above grammar subject to  $e \not =h$.  

The  ``small-step'' operational semantics for reducing a term in an environment $\E$ (a set of location names $\loc$ and store $\xs$, and a set of exception names $\ex$)  is given in  (Table \ref{ops2}). Variable names not occurring on the left of a rule are assumed fresh.
\begin{table}
\begin{center}
\begin{tabular}{l c l}
$E[\cse\spc \inj_i(V)\spc  \as \spc \inj_1(x).M_1|\inj_2(x).M_2],\E$&$ \lra$&$ E[M_i[V/x]],\E$ \\ 
$E[\match \spc  \langle U,V\rangle \spc \as\spc(x,y).M],\E$ &$\lra$&$ E[M[U/x,V/y]],\E$\\
$E[(\lambda x.M)\spc V],\E$&$ \lra$&$ E[M[V/x]],\E$ \\ $E[\lett \spc x = [V]\spc \inn\spc M],\E$ &$ \lra $&$ E[M[V/x]],\E$\\
$E[\new\spc ()],\E[(\loc)]$ & $\lra$ &$  E[[\langle \set(a),!a\rangle]],\E[\loc\cup\{a\}]$ \\
$E[\set(a) \spc V],\E[\xs] $&$\lra$&$ E[[()]],\E[\xs[a \mapsto V]]$ \\$E[!a\spc ()],\E[\xs]$ &$\lra$&$ E[[\xs(a)]],\E$\\ 
 $E[\nxn\spc ()],\E[\ex]$&$\lra$&$ E[[\langle \catch(e),\throw(e)\rangle]],\E[\ex\cup\{e\}]$ \\
$E[\catch(e)\spc \lambda x.E_e[\throw(e)\spc ()]],\E $&$\lra$&$ E[[()]],\E$ \\
$E[\callcc\spc V], \E$ &$\lra $ & $ E[V\spc \lambda x.\#E[x]],\E$\\
$E[\#(M)],\E $ &$\lra $ & $M,\E$

\end{tabular}
\caption{Operational Semantics of $\Lce$}\label{ops2}
\end{center} 
\end{table}
For a program (computation) $M:\com$, we write $M\Downarrow$ if $M,\varnothing$ reduces to $\ski$. Observational approximation and  equivalence are defined with respect to this  notion of convergence: $M \lesssim N$ if for all closing contexts, $C[\_]:\com$, $C[M] \Downarrow$ implies $C[N]\Downarrow$. $M \approx N$ if $M \lesssim N$ and $N \lesssim M$.

\subsection{Expressiveness}
We make some remarks on the expressiveness of our language. Although we have used a simplified version of exceptions   which  do not carry explicit  values, we may macro-express value-carrying  exceptions by using references to pass values through the store. For example, for any type $T$, define the type $\exn[T]$ of exceptions carrying values of type $T$ to be  $((\com \rightarrow \0) \rightarrow T) \times (T \rightarrow \0)$, so that applying  right-projection to a value raises an exception with that value, and applying left projection to a (thunked) computation captures an exception and returns the value it carries.  Then we may define an object declarating an exception of type $T$ --- $\nxn_T:\com \rightarrow \exn[T] =_{df}$
$$\lambda ().\new\spc a.\nxn\spc e. [\langle \lambda f.\catch\spc e \spc \inn \spc (f\spc ());\deref(a),\lambda x.(a:=x);\throw(e) \rangle]$$

We may represent ML or Java-style exception \emph{handling} --- i.e. including code to be run if only if a given exception is caught --- by using exceptions \emph{or} continuations to escape from the handler context if an exception is not raised, defining e.g.
 $${\mathtt{handle}} \spc e\spc\inn\spc N\spc \with\spc M   =_{df} \callcc(\lambda k. (\catch\spc e \spc \inn\spc  N;(k\spc ()));M)$$

By combining references, exceptions and continuations we may express \emph{resumable exceptions} which may return to the point at which they were raised. 
e.g. define the declaration ${\mathtt{resumable\_exn}}: ((1\rightarrow 0) \rightarrow (T \rightarrow  \emp)) \times (\com \rightarrow T)$ as follows: 
$$\new \spc a\spc \inn\spc \nxn\spc e\spc \inn \spc [\langle \lambda f. (\catch\spc e \spc \inn\spc  f\spc ());\deref(a), \callcc(\lambda k.a:=k;\throw(e))\rangle] $$
Right projection captures the current continuation and raises a (local) exception, left projection traps the exception and returns the continuation from the point it was thrown as a first-class function.


Finally, we note that exceptions and continuations are used in \cite{GRR} to implement prompts in Standard ML of New Jersey. Prompts are a form of locally declared, dynamically bound, delimited control operator which may be used to express   local exceptions, as defined here, and a \emph{delimited} form of $\callcc$. However, the implementation of prompts in SML$_{\mathrm{NJ}}$ uses global variables and is not therefore fully compositional: we leave a semantic investigation of the relationship between exceptions, continuations and delimited control as future work.  
\section{Denotational Semantics: Preliminaries}
First, we fix what we mean by a model of the type-theory $\L$ (essentially, a model of the computational $\lambda$-calculus \cite{MogC}): 
\begin{itemize}
\item a category $\C$ with finite, distributive  coproducts and products (including terminal and initial objects) and 
\item a strong monad $(\Sigma,\eta,\mu,\tau)$ on $\C$ such that for any $A$ and $B$ in $\C$, the exponential $A \Rightarrow \Sigma B$ exists. 
\end{itemize}
\emph{Types} are interpreted  as objects of $\C$  ---  $\[\com\]$ and  $\[\emp\]$ are the terminal and initial objects and  $\[S \times T\] = \[S\] \times \[T\]$, $\[S + T\] = \[S\] + \[T\]$ and  $\[S \rightarrow T\] = \[S\] \Rightarrow \Sigma \[T\]$. \\
For a context $\Gamma = x_1:T_1,\ldots, x_n:T_n$, define $\[\Gamma\] = \[T_1\]\times \ldots \times \[T_n\]$.  
\begin{itemize}
\item \emph{Values} $\Gamma \vdash_v V:T$ are interpreted as morphisms from $\[\Gamma\]$ to $\[T\]$.  
\item \emph{Computations} $\Gamma \vdash_c M:T$ are interpreted as  morphisms from $\[\Gamma\]$ to $\Sigma\[T\]$. 
\end{itemize}
Formal semantics  are given in Table 3.
\begin{table}

\begin{tabular}{l c l}
$\[\Gamma \vdash_v ():\com\]$&$ =$&$ t_{\[\Gamma\]}$\\
$ \[\Gamma,x:T \vdash_v x:T\]$&$ = $&$\pi_r$\\ 
$\[\Gamma \vdash_v \langle U,V \rangle:   S\times T\]$&$ =$ &  $\langle \[\Gamma \vdash_v U:S \], \[\Gamma \vdash V:T\]\rangle$\\
$\[\Gamma \vdash_v \inj_i(V):T_1+ T_2\]$&$ =$&$ \[\Gamma \vdash_v V:T_i\];\iota_i$\\
$\[\Gamma \vdash_v \lambda x.M:S \rightarrow T\]$&$ =$&$ \Lambda(\[\Gamma,x:T \vdash_c  M:T\])$\\
\\
$\[\Gamma \vdash_c \void(V):T\]$&$ = $&$\[\Gamma \vdash_v V:0\];i_{\Sigma\[T\]}$\\
$\[\Gamma \vdash_c [V]:T\]$&$ = $&$\[\Gamma \vdash_v V:T\];\eta$\\
 $\[\Gamma \vdash_c U\spc V:T\]$&$ =$&$ \langle \[\Gamma \vdash_v U:S \rightarrow T\], \[\Gamma \vdash_v V:T\]\rangle;\app$\\
$\[\Gamma \vdash_c\cse \spc V\spc \as\spc \inj_1(x).M|\inj_2(x).N\]$&$ =$&$ \langle \id_{\[\Gamma\]}, \[\Gamma \vdash_v V:S_1 +S_2\]\rangle;[\[\Gamma,x:S_1 \vdash_c M:T\],\[\Gamma,x:S_2 \vdash_c N:T\]]$   \\
 $\[\Gamma \vdash_v \match \spc V\spc \as (x,y).M:T\]$&$ = $&$\langle\id_{\[\Gamma\]}, \[\Gamma \vdash_v V:S_1 \times S_2\]\rangle;\[\Gamma,x:S_1,y:S_2\vdash_c M:T\]$  \\
 $\[\Gamma \vdash_c \lett \spc M = x\spc \inn\spc N:T\]$&$ =$&$ \langle \id_{\[\Gamma\]},\[\Gamma \vdash_c M:S\]\rangle;t_{\[\Gamma\],\[S\]};\[\Gamma,x:S \vdash_c M:T\]^*$\\  
\end{tabular}
\caption{Interpretation of $\L$-terms}
\end{table}

\subsection{Game Semantics}
We now review the game semantics of $\L$ and its extensions with references \cite{AHM} and continuations \cite{LT} (to which we refer for further details), in a category of arenas and thread-independent strategies.

An \emph{arena} $A$ is a bipartite labelled  forest --- a triple $\langle M_{A},\vdash_{A},\lambda_A \rangle$, where $M_A$ is the set of nodes (moves),  $\vdash_{A} \subseteq M_A \times M_{A}$ (the \emph{enabling} relation) is the set of edges, and $\lambda_A:M_{A} \rightarrow \{Q,A\}$ is a labelling function which partitions moves as  \emph{answers} (A) or \emph{questions} (Q), such that answers are enabled by questions. The set of root nodes of the forest is denoted $M^I_A$ --- these are called \emph{initial moves}.  
Partitioning of  $M_A$ into \emph{Player} and \emph{Opponent} moves may be inferred from the requirement that initial moves are Opponent moves, and that Player moves are enabled by Opponent moves and vice-versa.

Key constructions on arenas are:
\begin{itemize} 
\item The disjoint sum of forests (product): $A \times B = (M_A + M_B,\vdash_A + \vdash_B, [\lambda_A,\lambda_B])$. 
\item The graft of $A$ onto the roots of $B$ (function space): $A \Rightarrow B =\\ (\coprod_{m \in M^I_B}M_A + M_B, (\coprod_{m \in M^I_B} \vdash_A) + \vdash_B \cup \{(m,\inj_m(n))\ |\ m \in M^I_B,n \in M^I_A\},[[\lambda_A\ |\ m \in M^I_B],\lambda_B])$.  
\end{itemize}

A \emph{legal justified sequence} over the  arena $A$ is a finite alternating  sequence of moves of $A$ in which each occurrence of a  non-initial move $n$  comes with a unique \emph{justification pointer} to a preceding occurrence of a move $m$ which  enables $n$ (i.e. such that $m \vdash_A n$). 

A  strategy $\sigma$ over an arena $A$ is a non-empty, even-prefix-closed set of even-length alternating justified sequences over $A$, satisfying: 
\begin{description}
\item [Determinacy]If $sa,sb \in \sigma$ then $b=c$. 
\item [Thread-independence] If $r,s,t$ are even-length legal sequences such that $t$ is the interleaving of $r$ and $s$, then $t  \in \sigma$ if and only if $r,s \in \sigma$.
\end{description} 

The \emph{pending question prefix} (if any) of a justified sequence $s$ is the greatest prefix of $s$ ending with a question which does not occur between an answer and its justifying question in $s$: i.e.
\begin{itemize}
\item $\pending(sq) = q$
\item $\pending(sqta) = \pending(s)$, where $q$ justifies $a$.
\end{itemize} 
A strategy $\sigma$ is \emph{well-bracketed} if it satisfies the following condition:
\begin{description}  
\item  [Well-Bracketing]Any answer-move played by $\sigma$ is justified by the question pending when it was played --- i.e $sqta \in \sigma$ (where $a$ points to $q$) implies $\pending(sqt) = sq$. 
\end{description}

Composition of strategies $\sigma: A \Rightarrow B,\tau: B \Rightarrow C$ is by parallel composition plus hiding of moves in $B$. 
$\sigma;\tau =  \{s \in L_{A\Rightarrow B}\ |\ \exists t \in L_{((A \Rightarrow B) \Rightarrow C)}.t\restrict A,B \in \sigma \wedge t \restrict B,C  \in \tau \wedge t \restrict A,C = s\}$. This  yields a Cartesian closed category $\G$  in which objects are arenas, morphisms from $A$ to $B$ are strategies on $A \Rightarrow B$, and identities are copycat strategies. It has a cartesian closed, wide subcategory $\G^{B}$ in which morphisms are well-bracketed strategies  \cite{AHM,LT}. 

\subsection{Semantics of $\L$}
We interpret $\L$ by exhibiting a strong monad on the category of ``pre-arenas'' obtained by applying the $\Fam(\_)$ construction (small co-product completion) to $\G^B$ (following \cite{AMV}). For any category $\C$, $\Fam(\G)$ is the category of \emph{set-indexed families} of objects of $\C$, which has as morphisms from $\{A_i \ |\ i \in I\}$ to $\{B_j\ |\ j \in J\}$, a pair $\langle f:I \rightarrow J,\{\psi_i:A_i \rightarrow B_{f(i)}\ |\ i \in I\}\rangle$ of a re-indexing function and a family of morphisms in $\C$. 

If $\C$ is Cartesian closed, then so is $\Fam(\C)$:
\begin{itemize}
\item $\Fam(\C)$ has co-products, given by the disjoint union of indexed families. 
\item $\Fam(\C)$ has products, --- $\{A_i \ |\ i\in\ I\} \times \{B_j \ |\ j \in j\}$ is $ \{A_i \times B_j \ | \ \langle i,j\rangle \in I \times J\}$. 
\item $\Fam(\C)$ has exponentials --- in particular, for any arena $B$,  exponentials of the singleton family $\{B\}$:  $\{A_i\ |\ i \in I\} \Rightarrow \{B\} = \{\Pi_{i \in I}(A_i \Rightarrow B)\}$. 
\end{itemize}
A justified sequence on $A \Rightarrow B$ is \emph{linear} if every initial move in $B$ justifies exactly one initial move in $A$. A
(thread-independent) strategy $\sigma:A \rightarrow B$ is linear if it contains some  non-empty sequence, and every sequence $s \in \sigma$ is linear.  A \emph{tree arena} is an arena with a unique root (initial move).
Let $\G_S$ be the subcategory of $\G^B$ consisting of tree arenas and linear strategies.     
\begin{proposition}The inclusion of $\G_S$ in $\Fam(\G^B)$ has a left adjoint $\Sigma\_$.
\end{proposition}
\begin{proof}
The \emph{lifted sum} \cite{AMV} of a family of arenas $A = \{A_i \ |\ i \in I\}$ is the tree $\Sigma A$ with a single (question) root node, beneath which are answer nodes for each $i \in I$, beneath each of which is the arena $A_i$:
\begin{itemize}
\item$M_{\Sigma A} = (\coprod_{i \in I}(M_{A_i} + \{a_i\})) + \{q\}$
\item$\lambda^{QA}_{\Sigma A}(q) = Q, \lambda^{QA}_{\Sigma A}(a_i) = A \ i \in I,\ \ \lambda^{QA}_{\Sigma A}(\langle m, i \rangle) = \lambda_{A_i}(m)$
\item$* \vdash_{\Sigma A} q \vdash_{\Sigma A} a_i \vdash_{\Sigma A} \langle m,i\rangle$ and $\langle l,i\rangle \vdash \langle n, i\rangle$   where $m \in M^I_{A_i}$ and $l \vdash_{A_i} n$.
\end{itemize}

There is an evident correspondence between non-empty even-length  linear sequences on $A \Rightarrow \Sigma B$ and even-length sequences on $A \Rightarrow B$  yielding an  adjunction \begin{center}
{\Large $\G_S(\Sigma A,B) \over \Fam(\G^B)(A,\{B\})$}
\end{center} 
\end{proof}
Hence we have a (strong) monad on  $\Fam(\G^B)$ sending $A$ to the singleton family $\{\Sigma A\}$ \cite{AMV}, giving a semantics of $\L$. To extend this to a semantics of $\Lr$ it suffices to give the denotation of the non-functional part --- the  constant $\new_T: \var[T]$ --- as a  strategy
 $\cell_{A}:\Sigma(\Sigma A \times (A \Rightarrow \Sigma\1))$ defined in \cite{AHM}, which takes an argument of type $T$ and behaves as a reference cell initialized with that argument. This yields a computationally adequate semantics of $\Lr$, as proved in \cite{AHM}:
\begin{proposition}$M \Downarrow$ if and only if $\[M\] \not = \bot$. \end{proposition}
\subsection{Semantics of Continuations}
We may interpret call-with-current-continuation directly --- by identifying its denotation as a strategy, as  for new variable declaration --- but also indirectly, by \emph{CPS interpretation}. We recall the observation in \cite{LT} that relaxing the bracketing condition on strategies renders the lifted sum of games equivalent to a CPS construction. More precisely,  let   $U_C:\Fam(\G) \rightarrow \Fam(\G^B)$ be the identity-on-morphisms functor which acts on arenas by relabelling every answer as a question. Then there is an evident isomorphism of arenas: $U_C(\Sigma A) \cong  (\Sigma 0)^{(\Sigma 0)^{A}}$   yielding an equivalence of the lifting monad (on $\Fam(\G^C)$) to a CPS monad on $\Fam(\G)$ :
\begin{lemma}\label{cc} $\Sigma\cdot U_C \cong U_C\cdot (\Sigma 0)^{(\Sigma 0)^{\_}}$
\end{lemma}
Hence we have a strong monad $\Sigma\_$ on $\Fam(\G)$, equivalent to the one-move-game continuations monad --- and thus a semantics of $\L$ in $\G$, such that the inclusion functor $J:\G^B \rightarrow \G$ preserves meanings of $\Lr$ terms. 

For any game $A$, the isomorphism from $(A \Rightarrow \Sigma 0) \Rightarrow \Sigma 0$ into $\Sigma A$ yields a map  $\callc_{A,B}:((A \Rightarrow \Sigma B)   \Rightarrow \Sigma A) \rightarrow  \Sigma A$ with which to interpret $\callcc$. Concretely, this responds to the initial question ($\lbl$) with a  Player question ($\ok$),  and to its answer or the  subsequent question (${\mathsf{throw}}$)  with an answer to the initial question ($\caught$)  (note that this violates the bracketing condition), and  thereafter plays copycat between moves hereditarily enabled by  ($\ok$), and those hereditarily enabled by $\caught$. A typical play of $\callc_{\1,\0}$ is as follows: 

{$$\xymatrix@R=0pt@C=0pt{(\Sigma 0 & \Rightarrow & \Sigma\1) & \Rightarrow & \Sigma\1 \\ 
& & & & \lbl \\
& & \ok \ar@/^1pc/@{-->}[rru] \\
\jump  \ar@/^1pc/@{-->}[rru] \\
& &  & &\caught }$$

We may prove that our model is sound by using its  characterization as a CPS interpretation --- either by using an equational theory for CPS models to verify soundness of the reduction rules directly (see \cite{LT} for details), or by formalizing it as a CPS translation of $\Lc$ into $\Lr$ and showing that it is sound. We give a proof of the latter form in the Appendix, showing that:
\begin{proposition}Interpretation of programs as strategies in $\G$ is sound and computationally adequate: $M\Downarrow$ if and only if $\[M\] \not = \bot$. 
\end{proposition}

\section{Control Strategies}
We now extend the game semantics of $\Lc$ with continuations and exceptions, to interpret $\Lce$. We retain the interpretation of $\L$-types as (families of) arenas, but change the notion of strategy by  adding control-flow information to justified sequences in the form of ``control pointers''. (These were introduced in  \cite{Lli} in the context of interpreting a call-by-name language without first-class continuations: here, we relax the \emph{weak bracketing condition} imposed in that model.  
\begin{definition}\label{cs} A \emph{control sequence}  $s$ over an arena $A$ is an alternating justified  sequence $|s|$ over $A$, together with a \emph{control pointer} from each question  move  in $|s|$  either to a unique preceding question (or else to a distinguished root token $*$) --- such that Opponent moves point to  Player moves or $*$ and Player moves point to Opponent moves. 
\end{definition}
We write $C_A$ for the set of  control sequences over the  arena $A$. A \emph{control strategy} on $A$ is  a non-empty, even-prefix-closed set of even-length control sequences in $C_A$,  satisfying the determinacy and thread-independence conditions. 

In order to use our definition of composition for control strategies, we need to define the restriction operator on control sequences to replace ``dangling'' control pointers, by following back pointers  to hidden moves until an unhidden move is reached. Accordingly, we define the set of  \emph{open questions} of a control sequence  as follows:\\
$\open(\varepsilon) = \{\}$,\\
$\open(sqta) = \open(s)$, if $a$ is an answer\\ 
$\open(sqtq') = \open(sq) \cup \{sqtq'\}$ if $q'$ is a question with  a control pointer to $q$,\\ 
$\open(sq) = \{q\}$, if $q$ points to $*$.\\

We extend the restriction operation to control sequences by requiring that every move in  $s\restrict B$ points to the most recent preceding open move which is in $B$ (if any). 
With this definition of restriction, the original proofs of  well-definedness and associativity of parallel composition plus hiding  \cite{McCT} extend straightforwardly to control strategies.
 

To form a category, we also need to define identity morphisms (and other copycats) as control strategies.  Say that a control sequence  $s$ satisfies  (player) \emph{control locality} if every Player question in $s$ points to the pending question: let $\Loc_A$ be the set of control sequences over $A$ which satisfy this condition. Given a  strategy on $A$, we may define a local control strategy $\widehat{\sigma}$  on $A$ by taking all player-local sequences which correspond to sequences in $\sigma$ when pointers are ignored  i.e. $\widehat{\sigma} =  \{s \in Loc_{A}\ |\ |s| \in \sigma\}$. (In other words, by decorating sequences in $\sigma$ by adding control pointers from each Opponent question to some Player question, and from each Player question to its pending  question.)
 We define the identity control strategy to be $\widehat{\id_A}$ (and similarly for the other copycat strategies giving cartesian closed structure).   So  we may define a cartesian closed category $\CG$ in which objects are arenas, and morphisms from $A$ to $B$ are  control strategies on $A \rightarrow B$. 

We also observe that:
\begin{itemize}
\item The operation $\widehat{\_}$ is not functorial: the  arenas $\Sigma\1$ and $\Sigma\0 \Rightarrow \Sigma\0$ are isomorphic in $\G$ but not in $\CG$: the  composition of the images of these isomorphisms under $\widehat{\_}$ is not $\widehat{\id_{\Sigma\0 \Rightarrow \Sigma\0}}$.
 \item $\widehat{\_}$ \emph{is} functorial on well-bracketed strategies:  there is a faithful, identity-on-objects functor $J:\G^B \rightarrow \CG$ sending $\sigma:A \rightarrow B$ to $\widehat{\sigma}$.
\end{itemize}

\subsection{Semantics of Exceptions}
$\CG$ has structure with which to  model $\L$ --- (it is a CCC with a strong lifted sum monad $\Sigma\_$ on $\Fam(\CG)$. The functor $J:\G^B \rightarrow \CG$ preserves all of this structure, and hence the meaning of $\L$-terms. We interpret new reference declaration, and  call-with-current-continuation by decorating the corresponding underlying strategies with control pointers: i.e.
\begin{itemize}
\item $\[\new_T\]_C = \widehat{\cell_{\[T\]}}$
\item $\[\callcc\]_C = \widehat{\callc}$ 
\end{itemize}  
Since $\cell$ is well-bracketed, it lies within the domain of the functor  $J$, which therefore preserves the meaning of $\Lr$-terms by definition. 
 As we have noted, $\widehat{\_}$ is not functorial on non-well-bracketed strategies, and hence does not preserve the meaning of $\Lc$-terms in general.

\paragraph{Exceptions}
We interpret new-exception declaration as a  new-exception strategy $\xcell_C:\Sigma((\Sigma 0 \Rightarrow \Sigma\1) \times (\Sigma 0))$. This was   defined  in a weakly-bracketed setting in \cite{Lli}: it   relies  on control pointers to determine the current exception handler.
Its behaviour can be described as follows:
\begin{itemize}
\item Answer the initial (Opponent) question.
\item If Opponent plays  the initial move on the left  ($\mathsf{try}$)  then respond with a question ($\ok$).
\item If Opponent plays the initial move on the right ($\thro$), then answer the most recent open initial question on the left ($\caught$).  If there are no such moves,  do nothing, representing divergence caused by an uncaught exception.
\end{itemize}
In other words, $\xcell_C$ consists of the legal control sequences of the form $qa(((\mathsf{try}\ok)^*(\thro\caught)^*)^*$ such that each $\caught$ move is justified by the most recent open $\mathsf{try}$ move. Here is a typical play: 
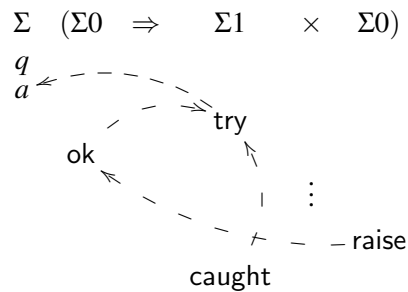
\begin{center}
$$\xymatrix@R=-1pt@C=5pt{ \Sigma & (\Sigma\0 & \Rightarrow & \Sigma\1 & \times & \Sigma\0) \\ 
q\\
a\\
&& & {\mathsf{try}}  \ar@/_1pc/@{-->}[lllu] \\
&\ok \ar@/^1pc/@{-->}[rru] \\
& & &  & \vdots \\
&& & & & \raze  \ar@/^1pc/@{-->}[lllluu] \\
& & & \caught \ar@/_1pc/@{-->}[uuuu]}$$
\end{center}
Proving soundness for this model directly is difficult due to the implicit nature of its representations of continuations and exceptions. Instead, we will provide  an alternative characterization of $\xcell$ by relating it to an exceptions monad.

\section{A Monadic Effect Semantics}
In this section, we construct  a semantics of $\Lce$ corresponding to  exceptions and continuation-passing-style interpretation. This has the advantage of relating these control effects to well-understood structure, at the cost of a less direct and concrete interpretation of terms.  Observe  that since  the  (strong) exceptions monad $\_ + E$ distributes over any other monad, we have a monad $\Sigma(\_ +1)$ on $\Fam(\G)$. By Lemma \ref{cc}, this is equivalent to the  exceptions-and-continuations monad $\R^{\R^{\_ + E}}$, for the answer object $\R = \Sigma0$.

In fact, we will describe the semantics of $\Lce$ in a category of ``exception arenas'' in which raising and propagating the global exception is represented by playing explicit ``exception moves''. The lifting monad in this model is equivalent to the above exceptions-and-continuations monad, giving a route to establishing its soundness, whilst it may be related to the control games interpretation by replacing runs of exception moves with control pointers.
\begin{definition}
An exception arena is a pair $(A,e_A)$ of an arena together with a function $e_A:\{m \in M_A\ |\ \lambda(m) = Q\} \rightarrow \{m \in M_A\ |\ \lambda(m) = A\}$ associating each question $q$ with a unique ``exception answer'', which is a child of $q$  --- i.e  $q \vdash e_A(q)$ --- and a leaf of $A$.  
\end{definition}
Exception moves correspond to a single, global exception: playing an exception answer in response to a non-exception move corresponds to \emph{raising} this exception, playing the pending exception answer in response to an exception move corresponds to \emph{propagating} it, and playing  a non-exception move in response to an exception move corresponds to \emph{handling} it.

Observe that if $(A,e_A)$ and $(B,e_B)$ are exception arenas, then $(A \Rightarrow B,\coprod_{m \in M^I_B} e_A +e_B$ and $(A \times B, e_A + e_B)$ are exception arenas. Hence we may define:
\begin{itemize}
\item a Cartesian closed category $\G^E$, which has exception-arenas as objects and unbracketed strategies on $A \Rightarrow B$ as morphisms from $A$ to $B$. 
\item A fully faithful (identity on morphisms)  cartesian closed functor $U_E:\G^E \rightarrow \G$ which forgets the exception answer labelling.
\end{itemize}
Given a family of exception-arenas  $A = \{(A_i,e_i) \ |\ i \in I\}$, define $\Sigma_E A = (\Sigma(\{A_i \ |\ i \in I\} +1), \coprod_{i \in I}e_i \cup \{q,\inr(a)\})$ --- i.e. the exception arena given by extending $\Sigma\{A_i \ |\ i \in I\} $ with an additional exception answer move $\inr(a)$ to the initial question.
 By definition, we have:
\begin{lemma}\label{ff}$U_E(\Sigma_E B) =  \Sigma (U_E(B) + 1)$. 
\end{lemma}
Hence by full faithfulness of $U_E$, we therefore have a strong monad $\Sigma_E$ on $\Fam(\G^E)$ such that $U_E\cdot \Sigma_E = \Sigma(\_ +1) \cdot U_E$, giving a semantics for the type-theory $\L$ in $\G^E$. Note that by Lemma \ref{cc}, $\Sigma_E$ is equivalent to the exceptions-with-continuations monad $\R^{\R^{\_ + 1}}$, where $\R$ is the one-move game --- i.e. $ U_C\cdot U_E \cdot \Sigma_E \cong \R^{\R^{\_ + 1}}  \cdot U_C \cdot U_E$.

We may interpret continuations and references in this model by adding exception-answers to our existing denotations for these features.    
Given an exception-arena $A$, let $K(A)$ be the arena obtained by erasing all of the exception-answers in $A$ --- i.e. $M_{K(A)} = M_{A} - \{e(q)\ |\ \lambda(q) = Q\}$. Say that an even-length  legal sequence $s$ on $A$ is exception-local if:
\begin{itemize}
\item Player always propagates exception moves: if $te(q)m \sqsubseteq s$ is even-length, then $m = e(q')$, where $q'$ is the pending question of $te(q)$.
\item Player never raises an exception: if $tmn \sqsubseteq s$ is even-length, and $m$ is not an exception-move, then $n$ is not an exception move.
\end{itemize}
Given an exception-local sequence, let $\overline{s}$ be the legal sequence on $K(A)$ obtained by erasing exception moves in $s$. Let the \emph{exception-completion} of a strategy $\sigma$ on $K(A)$,  be the set $\widetilde{\sigma}$  of exception-local sequences on $A$ such that $\overline{s} \in \sigma$. (Note that this is deterministic and thread-independent.) Hence we may define the denotation of $\callcc$ and $\new$ as the exception-completions of their denotations in $\G$.

 To interpret exception declaration and handing, first note that we have evident \emph{raise} and \emph{handle} strategies on $1 \rightarrow \Sigma_E 1$ and $\Sigma_E 0  \rightarrow \Sigma_E 1$ which raise and handle the \emph{global} exception by playing the exception-answer to the initial question and  respond to Opponent's playing of an exception answer by playing a ``non-exception-answer'', respectively. The interpretation of  $\nxn:((\1 \rightarrow 0) \rightarrow \1) \times (\1 \rightarrow \0)$ as a strategy  $\xcell_E: ((\Sigma_E \0 \Rightarrow \Sigma_E \1) \times  \Sigma_E\0$ combines these behaviours with local state: the handle method handles the global exception if the raise method has been used to raise a global exception which has not yet been handled, and propagates it otherwise. Formally, define $\xcell_E$ to consist of all plays of the form $qa ({\mathsf{handl}}\ok)^* (\raze\cdot e[\raze]))^*(e[\ok]\caught)^*$ which are Player-well-bracketed. Example plays are given in Figure \ref{exmpl}. 
\begin{figure}
\begin{center}
\begin{tabular}{c c c}
{$$\xymatrix@R=-1pt@C=5pt{ \Sigma & (\Sigma_E\0 & \Rightarrow & \Sigma_E\1 & \times & \Sigma_E\0) \\ 
q\\
a\\
&& & {\mathsf{try}}   \\
&\ok  \\
& & &  & \vdots \\
& & & & & \raze   \\
& & & & &  e(\raze)\\
&  e(\ok)\\
& & & \caught}$$} &\ \ \ \ \  &
{$$\xymatrix@R=-1pt@C=5pt{ \Sigma_E & (\Sigma_E\0 & \Rightarrow & \Sigma_E\1 & \times & \Sigma_E\0) \\ 
q\\
a\\
& & & {\mathsf{try}}  \\
&\ok  \\
& & &  & \vdots \\
&  e(\ok)\\
& & & e({\mathsf{try}})}$$}  
\end{tabular}
\end{center}
\caption{Plays of $\xcell_E$}\label{exmpl}
\end{figure}

\subsection{Soundness}
We may prove soundness and adequacy of our exceptions monad model by showing that it corresponds to an \emph{exception-passing-style} translation  $(\_)^E$  from $\Lce$ to  $\Lc$. This acts on types as follows:
\begin{itemize}
\item $\emp^E = \emp$, $\com^E = \com$,
\item $(S \times T)^E = S^E \times T^E$,
\item $(S + T)^E = S^E + T^E$,
\item $(S \rightarrow T)^E = S^E \rightarrow (T^E + \com)$.
 \end{itemize}

Translation of computations $x_1:S_1,\ldots,x_n:S_n \vdash M:T$ as computations $x_1:S_1^E,\ldots,x_n:S^E_n \vdash_c M^E:T^E + \com$ and values $x_1:S_1,\ldots,x_n:S_n \vdash V:T$ as values  $x_1:S_1^E,\ldots,x_n:S^E_n \vdash_E V^E:T^E$  is given in Table 4. The translation of $\Lc$  arises straightforwardly from  the strong monad structure, so we focus on the interpretation of \emph{local} exceptions --- specifically, the new-exception declaration. Right injection and pattern matching may be used to define evident operations  which raise and handle the single \emph{global} exception, respectively. But how can we use these to construct an object whose methods raise and handle a \emph{local} exception? The answer is: we use the local state in our underlying model/metalanguage. Our exception has as its  internal state a Boolean variable $e$ acting as a flag. The  raise method for the local exception object sets $e$ and raises the global exception. The handle method for the object handles the global exception, tests $e$ and resets it if it is set (i.e. $e$ had been raised and has now been handled) or re-raises the global exception if it is not set (i.e. some other exception was raised and now needs to be propagated). So we have two methods:
\begin{itemize}
\item $\rais(e) = \lambda z.e:=\tru;\inj_2(()):(\com \rightarrow \emp)^E $  
\item ${\mathtt{handle}}(e) = \lambda f.(f\spc ());\If\spc \deref(e)\spc \then \spc (e:=\ff;\inj_2(())) \spc \el\spc \inj_2(()):((\com \rightarrow \emp) \rightarrow \com)^E$
\end{itemize}
New-exception declaration simply aggregates these methods and hides the internal state $e$.
\begin{table}

\begin{itemize}
\item $(x)^E = x$ 
\item  $[V]^E = [\inj_1(V^E)]$ 
\item $(\lett\spc M = x\spc \inn\spc N)^E = \lett\spc M^E = y\spc \inn \spc \cse\spc y\spc \as\spc \inj_1(x).N^E|\inj_2(x).\inj_2(())$
\item  $(\lambda x.M)^E = \lambda x.M^E$,\ \ \  $(U\spc V)^E = U^E \spc V^E$ 
\item $(\match\spc V\spc \as\spc (x,y).M)^E = \match\spc (x,y)\spc \as\spc V^E\spc \inn \spc M^E$, \ \ \ \ $\langle U,V\rangle^E = \langle U^E,V^E\rangle$,
\item  $()^E = ()$,\ \ \ \ $\void(V)^E = \void(V^E)$
\item $\inj_i(V)^E = \inj_i(V^E)$, \ \ 
 $(\cse\spc V\spc \spc \as\spc \inj_1(x).M|\inj_2(x).N)^E = \cse\spc V^E\spc\as\spc \inj_1(x).M^E|\inj_2(x).N^E $ 
\item $\callcc(V)^E = \callcc(\lambda k.V^E\spc \lambda x.k\spc \inj_1(x))$  
\item $\new^E = \lambda ().\new\spc a\spc \inn\spc [\inj_1(a)]$ 
\item $\nxn^E =\lambda ().\new\spc x:=\ff.[\inj_1(\langle {\mathtt{handle}}(e),\rais(e)\rangle)]$
\end{itemize}
\caption{Exception-passing translation}
\end{table}
We establish that the interpretation of local exceptions via translation into $\Lc$ is sound.
\begin{lemma}$M \Downarrow$ if and only if $M^E \Downarrow$.  
\end{lemma}
\begin{proof}
The key step is to  break the exception-propagation rule down to propagate exceptions past each  $\lett$-context and non-matching handler.  
\end{proof}
We have already noted that the forgetful functor $U_E$ sends $\Sigma_E$ to the (lifted) exceptions monad $\Sigma(\_ + 1)$. This correspondence extends  to the interpretation of $\Lc$-types and terms.
\begin{lemma}For any  $\L$-term, $\Gamma \vdash M:T$, $U_E(\[\Gamma \vdash M:T\]) = \[M^{E}\]$.
\end{lemma}
\begin{proof}For types, and terms of $\Lc$, this follows from equivalence of $\Sigma_E$ and $\Sigma\_ +1$ via $U_E$. Thus, it remains to observe that these interpretations agree on the interpretation of new exception declaration --- i.e.
$\[\lambda ().\new\spc x:=\ff.[\inj_1(\langle \handl(e),\rais(e)\rangle)]\] = U_E(\xcell_E)$.
\end{proof}
Hence we have soundness and adequacy for the exception-arena semantics. 
\begin{proposition}$M \Downarrow$ if and only if $\[M\]_E \not=\bot$  
\end{proposition}

\section{Relating Control Games to the Exception Monad} 
We now relate the monadic and  control-games interpretations of $\Lce$ by  establishing a correspondence between the two models: a meaning-preserving functor into $\CG$ from a subcategory of  $\G^E$ consisting of exception-arenas and \emph{exception-propagating strategies}. 

 A sequence $s$ over $A$ is \emph{exception-propagating} if whenever Opponent raises an exception, Player always propagates it by playing the exception-answer to the pending question, and vice versa. Formally, define the set of exception-propagating sequences to be the least set of justified sequences such that:
\begin{itemize}
\item The empty sequence is exception-propagating, 
\item If $s$ is exception-propagating, and $m$ is not an exception move, then $sm$ is exception-propagating.
\item If $s$ (of even length) is exception-propagating then $se(q)e(q')$ is exception-propagating, where $q'$ is the pending question of $se(q)$.
\end{itemize}
 We write $EP_A$ for the set of exception-propagating control sequences on $A$.
 Recall that we defined  $K(A)$ to be the arena obtained by erasing all of the exception-answers in $A$. Given $s \in EP_A$ we define a control sequence $K(s)$ on $K(A)$ by:
\begin{itemize}
\item First, adding a control pointer from each question to its  pending question (or the token $*$ otherwise).
\item Then, deleting all exception answers.   
\end{itemize}
This is a well-defined control sequence; it is alternating since if $s$ is exception-propagating then all exception-moves  in $s$ come in adjacent pairs. Control pointers alternate in polarity  since the pending question is always of opposite polarity to the move about to be played (and there is always a pending question at any Player move). For an example, the first typical play given for the  new-exception strategy $\xcell_E$ (Fig. 1) is transformed to the typical play given for the corresponding control strategy $\xcell_C$

Extend the definition of $K$ to all justified sequences on the exception-arena  $A$ by letting $K(s) = K(t)$, where $t$ is the greatest exception-propagating prefix of $s$. A \emph{strategy} $\sigma$ on $A$ is \emph{exception-propagating} if $K(\sigma) = \{K(s)\ |\ s\in \sigma\}$ is a well-defined (thread-independent) control strategy. In other words:
\begin{itemize}
\item $K(\sigma)$ consists of even-length sequences --- $\sigma$ always propagates exceptions raised by Opponent. 
\item $K(\sigma)$ is even-branching --- $\sigma$ ignores exceptions raised by Opponent once they have been handled (but not their effect on the exception handling context). 
\end{itemize}  
We show that this is a compositional property of strategies, and that the action of $K$ is functorial, based on the following lemma:
\begin{lemma}Given  $s \in L_{(A \rightarrow B) \rightarrow C}$, if $s \restrict A,B$ and $s \restrict B,C$ are legal and exception-propagating, then:
\begin{itemize}
\item  $s \restrict A,C$ is exception-propagating. 
\item $K(s)\restrict A,C = K(s \restrict A,C)$.  
\end{itemize}
\end{lemma}
\begin{proof}Since the $s \restrict A,B$ and $s \restrict B,C$ are exception-propagating, any runs of exception-answers occur in even-length blocks in $s$, and erasing the part in $B$ leaves an exception-propagating sequence.

We show by induction on the length of $s$ that the pending question in $s \restrict A,C$ is the pending question in the relevant fragment $s \restrict A,B$ or $s\restrict B,C$, and so control pointers in $K(s)\restrict A,C$ and  $K(s \restrict A,C)$ agree.  
\end{proof}Based on this lemma, we show:
\begin{proposition}The composition of exception-propagating strategies is exception-propagating. 
\end{proposition}

It is straightforward to verify that the identity strategy is exception propagating, with $K(\id) = \id$. Hence:
\begin{itemize}
\item Exception-propagating strategies  form a lluf subcategory $\G^{EP}$ of $\G^E$. 
\item $K$ acts as a functor from $\G^{EP}$ to $\CG$.
\end{itemize}
 Evidently, $K$  preserves Cartesian closed structure, and $\Sigma_E \cdot K \cong \Sigma \cdot K$. So for $\L$-types we have $K(\[T\]_{E}) \cong \[T\]_C$. Moreover, if we  apply $K$ to the exception-completion of a strategy, this is equivalent to decorating with control-pointers to the pending moves --- i.e. $K(\widetilde{\sigma}) = \widehat{\sigma}$. So $K$ preserves the meaning of $\Lc$-terms. It remains to check that this is the case for exception declaration.
\begin{lemma}$\xcell_E$ is exception-propagating, and  $K(\xcell_E) \cong \xcell_C$. 
\end{lemma}
\begin{proof}
Recall that $\xcell_E$ is the set of well-bracketed plays $qa ({\mathsf{handle}}\ok)^* (\raze\cdot e[\raze]))^*(e[\ok]\caught)^*$. This is evidently exception-propagating, and erasing the exception moves on exception-propagating plays leaves a control sequence  of the form $qa(((\mathsf{try}\ok)^*(\thro\caught)^*)^*$, where  the pending $\mathsf{try}$ move in the original sequence becomes the target of a control pointer and hence the most recent open handler as required. 
\end{proof}
Thus we have shown that:
\begin{proposition}Every $\Lce$-term $M$ denotes an exception-propagating strategy such that $K(\[M\]_{E}) = \[M\]_{C}$.  
\end{proposition}
and hence established soundness and adequacy for the exception-arena semantics. 
\begin{proposition}$M \Downarrow$ if and only if $\[M\]_C \not=\bot$.  
\end{proposition}

\section{Full Abstraction}
Having proved soundness for the control strategy model using algebraic methods (correspondence with the monad model), we may establish full abstraction  via reduction (by factorization) to the definability result for the original model of $\Lr$. 

Recall that a (thread-independent) strategy $\sigma$ is compact (in the inclusion order) if the set of \emph{well-opened} sequences (those having a unique initial  move) in $\sigma$ is finite. The following result is proved (by factorisation) in \cite{AHM}.
\begin{proposition}For any type $\L$-type $T$, every compact well-bracketed strategy on $\Sigma\[T\]$ is the denotation of a $\Lr$ term $M_\sigma:T$.
\end{proposition}
Since $J:\G^B \rightarrow \CG$ preserves the meaning of $\Lr$-terms, every compact strategy in the image of $J$ is definable --- i.e. all compact, local well-bracketed strategies which also satisfy the following condition:
\begin{description}
\item [Control blindness]$|\sigma| = \{|s|\ |\ s \in \sigma\}$ is a deterministic strategy on $A$.
\end{description}
\begin{corollary}\label{ldef}Every  compact local, well-bracketed and  control-blind  strategy over an $\L$ type-object is definable as a term of $\Lr$. 
\end{corollary}
We now factorize any compact control-blind strategy into the composition of a definable strategy with the denotation of $\callcc$. 
\begin{lemma}For any (compact) control-blind strategy  $\sigma:\1 \rightarrow A$ there is a (finitary) local, well-bracketed  strategy $\widetilde{\sigma}: ((\Sigma\0 \Rightarrow \Sigma\1) \Rightarrow \Sigma\1) \rightarrow A$ such that  $\Lambda(\callc_{1,0});\widetilde{\sigma} = \sigma$.
\end{lemma}
\begin{proof}This is essentially the factorization given in \cite{L}: we define a map $\widetilde{\_}$ from control sequences to Player well-bracketed sequences in $L_{((\Sigma\0 \Rightarrow \Sigma\0)  \Rightarrow \Sigma\1)  \Rightarrow A}$ which interjects $\lbl,\ok$ after each Opponent question, and blocks of $\jump,\caught$ moves before each Player answer, so that  any intervening questions are closed.



\end{proof}
So it remains to show that control-blind strategies may be factorized as the composition of a local control strategy with the strategy $\xcell$. 
\begin{lemma}For any (compact) control strategy $\sigma:\1 \rightarrow A$ there is a (compact) control-blind strategy $\widehat{\sigma}: \[\exn\] \rightarrow A$ such that  $\xcell;\widehat{\sigma} = \sigma$
\end{lemma}
\begin{proof}
We define a map $\widehat{\_}$ from control sequences on $A$ to control sequences on $\exn \rightarrow A$ which makes all control pointers on $O$-moves manifest. It raises an exception before each Opponent move, and handles it  after each $O$-move. The handler which catches the raised exception is determined by the control pointers from $O$-moves in $s$, so that  $|\widetilde{s}| = |\widetilde{t}|$ implies  $s = t$. Since  $\widetilde{s} \restrict A = s$, and   $\widetilde{s} \restrict \[\exn\] \in \xcell$, we have  $\xcell;\widehat{\sigma} = \sigma$ as required
\end{proof}
For any compact strategy, $\sigma:I \rightarrow \Sigma\[T\]$, the control-blind strategy $\widetilde{\sigma}$ is definable as a term  $x:\exn \vdash M$ and thus $\sigma$ is definable as $\lett \spc x= \nxn\spc \inn\spc M$. The proof of full abstraction based on finite definability is now standard. 
%

\begin{theorem}The model of $\Lce$ in $\Fam(\CG)_{\Sigma}$ is (inequationally) fully abstract --- i.e. for any $\L$-type $T$ and any terms $M,N:T$, $\[M\]_C \subseteq  \[N\]_C$ if and only if $M \lesssim N$.
\end{theorem}

\section{Conclusions and Further Directions}
\paragraph{Model checking exceptions}
Giving different representations of exceptions in games models may be useful in the developing field of program-verification based on semantic games. For example, we may observe that the set of exception-propagating sequences over a finite alphabet (with a specified subset of distinguished exception tokens) is regular, giving a way of extending results characterizing finite-state representable fragments of imperative languages to include local exceptions. On the other hand control pointers describe control flow (and, in particular,  exception handling points) directly, and so adding them to  game semantic approaches to control flow analysis \cite{MH} offers the possibility of reasoning about e.g. exception safety.
\paragraph{Delimited Control}  Further instances of delimited continuations such as locally declared, dynamically bound \emph{prompts} \cite{GRR} could be modelled by a similar analysis relating CPS interpretation to the stateful behaviour in games models.
\paragraph{Good Variables} Languages such as ML and Java have explicit exception types, so that an object of exception type must behave as an exception, whereas there is clearly no such constraint on  objects of the product type which we have used as an exception type. Extending our full abstraction results to such languages is liable to require some characterization of such behavioural constraints. This problem is analogous to the ``good variable'' problem for references, and we may look to research in this area for approaches to model ``good exceptions''  \cite{ntex}. Implementing \emph{wildcard handling} (e.g. Java's $\mathtt{finally}$) becomes straightforward when exceptions are passed as names through an exceptions monad, although a wildcard handler typically cannot trap an exception and then discover its name, and so a model should reflect this constraint. 

\bibliography{names}
\section*{Appendix: Soundness for $\Lc$ via CPS translation }
We give an interpretation of $\Lc$ in $\Lr$ --- a  CPS translation corresponding to the action of the CPS monad on our denotational model.  
This acts on types as follows:
\begin{itemize}
\item $0^C = 0$, $1^C = 1$,
\item $(S + T)^C = S^C + T^C$,
\item $(S \times T)^C = S^C \times T^C$,
\item $(S \rightarrow T)^C = (S^C \times (T^C \rightarrow \emp)) \rightarrow \emp$.
\end{itemize}
Values $x_1:S_1,\ldots,x_n:S_n \vdash_v V:T$ are translated as values $x_1:S_1^C,\ldots,x_n:S^C_n \vdash_v V^C:T^C$ and computations $x_1:S_1,\ldots,x_n:S_n \vdash_c M:T$ are translated as \emph{values} $x_1:S_1^C,\ldots,x_n:S^C_n \vdash_v M^C:(T^C \rightarrow 0)\rightarrow 0$ as defined in Table 3.

\begin{table}
\begin{itemize}
\item $(x)^C = x$,\ \ \  
$(\lett\spc x = M\spc \inn\spc N)^C = \lambda \kappa.M^C\spc \lambda m.\lett x =m \spc \inn\spc (N^C\spc \kappa)$
\item  $[V]^C =  \lambda \kappa.\kappa\spc V^C$ 
\item  $(\lambda x.M)^C = \lambda z.\match\spc (x,\kappa) \spc \as\spc x\spc \inn\spc   M^C\spc \kappa $,\ \ \ \  $(U\spc V)^C = U^C\spc V^C$
\item $\langle U,V\rangle^C = \langle U^C,V^C\rangle$ \ \ \ \ $(\match\spc (x,y)\spc \as\spc V. M)^C = \lambda \kappa.\match\spc (x,y)\spc \as\spc V^C. M^C\spc \kappa$ 
\item $()^C = ()$,\ \ \  $\void(V)^C = \lambda \kappa.\kappa\spc \void(V^C)$.
\item  $\inj_2(U)^C = \inj_2(V^C)$,\  $\inj_1(V)^C = \inj_1(V^C),$\\
$(\cse\spc V\spc \spc \as\spc \inj_1(x).M|\inj_2(x).N)^C = \cse\spc V^C\spc \as\spc \inj_1(x).M^C|\inj_2(x).N^C $ 
\item $\new^C = \lambda\kappa.\new\spc a\spc \inn\spc \kappa\spc \langle \lambda x.\fst(x)\spc (a:= \snd(x)), \lambda y. \snd(y)\spc \deref(a)\rangle$ 
 \item  $\callcc(V)^C = \lambda \kappa. V^C\spc \langle k,\lambda x.\match\spc (y,z)\spc \as\spc x \spc \inn \spc k\spc x\rangle$   
\end{itemize}
\caption{CPS translation of exception-free terms}
\end{table}
Extending to ${\mathcal{\L_{C}^{\#}}}$ by setting  $\#(M)^C = \lambda \kappa.M^C\spc \tau$ (where $\tau:1 \rightarrow 0$ is a variable representing the top-level continuation), we may show that reduction of a term tracks that of its translation: 
\begin{proposition}For any program $M:\com$, $M \Downarrow$ if and only if $M^C\spc \tau  \longrightarrow \tau\spc ()$  
\end{proposition}
We note that CPS interpretation corresponds (up to isomorphism) to the interpretation of $\Lc$-types and terms:
\begin{proposition}For any $\L$-type $T$, there is an isomorphism of arenas $\phi_T:U_C(\[T\]) \cong \[T^{C}\]$  such that  for any $\Lc$-term, $\Gamma \vdash M:T$, $U_C(\[\Gamma \vdash M:T\]);\phi_T = \[M^{C}\]$.
\end{proposition}
\begin{proof}The key type-constructor is the function type: we have $U_C(\[S \rightarrow T\]) = U_C(\[S\]) \Rightarrow  U_C(\Sigma \[T\]) \cong \[S^{C}\] \Rightarrow \R^{\R^{\[T^{C}\]}} = (\[S^{C}\] \times \[(T^{C}\]) \Rightarrow \R) \Rightarrow \R  \cong \[(S^{C} \times (T^{C} \rightarrow \emp)) \rightarrow \emp\] \cong (S \rightarrow T)^{C}$.    
\end{proof} Hence interpretation in the games model is sound and adequate:
\begin{proposition} $M\Downarrow$ if and only if $\[M\] \not = \bot$. 
\end{proposition}
\end{document}